\documentclass[american,aps,prl,twocolumn,superscriptaddress]{revtex4}
\usepackage[T1]{fontenc}
\usepackage[latin9]{inputenc}
\usepackage{babel}
\usepackage{amsmath}
\usepackage{mathrsfs,bm,amssymb,amsbsy,graphicx,amsthm}
\usepackage{hyperref}
\usepackage{times}
\usepackage{braket}

\makeatletter

\theoremstyle{plain}
\newtheorem{thm}{\protect\theoremname}

\makeatother

\providecommand{\theoremname}{Theorem}

\newcommand{\tr}{{\rm Tr}}

\newcommand{\proj}[1]{\ket{#1}\bra{#1}}

\begin{document}

\title{Are general quantum correlations monogamous?}

\author{Alexander Streltsov}
\affiliation{Heinrich-Heine-Universit\"at D\"usseldorf, Institut f\"ur Theoretische Physik III, D-40225 D\"usseldorf, Germany}

\author{Gerardo Adesso}
\affiliation{$\mbox{School of Mathematical Sciences, University of Nottingham, University Park, Nottingham NG7 2RD, United Kingdom}$}

\author{Marco Piani}
\affiliation{$\mbox{Institute for Quantum Computing and Department of Physics and Astronomy, University of Waterloo, Waterloo ON
N2L 3G1, Canada}$}

\author{Dagmar Bru{\ss}}
\affiliation{Heinrich-Heine-Universit\"at D\"usseldorf, Institut f\"ur Theoretische Physik III, D-40225 D\"usseldorf, Germany}

\begin{abstract}
Quantum entanglement and quantum non-locality are known to exhibit monogamy, that is, they obey strong constraints on how they can be distributed among multipartite systems.
Quantum correlations that comprise and go beyond entanglement are quantified by, e.g., quantum discord. It was observed
 recently that for some states quantum discord is not monogamous. We prove in general that any measure of
correlations that is  monogamous for all states and satisfies reasonable basic properties
must vanish for all separable states: only entanglement measures can be strictly monogamous.
Monogamy of other than entanglement measures can still be satisfied for special, restricted cases:
we prove that the geometric measure of discord satisfies the monogamy inequality
on all pure states of three qubits.
\end{abstract}

\maketitle

Entanglement, nonclassical correlations, and nonlocal correlations,
are all forms of correlations between two or more subsystems of a composite
quantum system that are different from strictly classical correlations,
and in general different from each other. One of the characteristic
traits of classical correlations is that they can be freely shared.
A party $A$ can have maximal classical correlations with two parties
$B$ and $C$ simultaneously. This is no longer the case if quantum
entanglement or nonlocal correlations are concerned \cite{Terhal2004}.
The limits on the shareability of those types of nonclassical correlations are known as monogamy constraints, see Fig. \ref{fig:1} for illustration. Strict monogamy inequalities
have been proven that constrain the distribution of particular measures
of entanglement and nonlocal correlations (the latter expressed in
terms of violation of some Bell-type inequality \cite{Clauser1969})
among the subsystems of a multipartite system \cite{Coffman2000,Osborne2006,Adesso2006,Hiroshima2007,Toner2009,Barrett2006,Pawlowski2009,Koashi2004,Seevinck2010}.
These relations can be seen as a particular case of trade-off relations
that in general may relate and constrain different quantifiers of
correlations \cite{Koashi2004,Horodecki2007}. Monogamy
is the crucial property of correlations that makes quantum key distribution
secure \cite{Terhal2004,Pawlowski2010}, even in no-signalling theories
more general than quantum mechanics.

Nonclassical correlations that go beyond entanglement, often quantified e.g.~via the quantum
discord \cite{Ollivier2001,Henderson2001}, have recently
attracted considerable attention \cite{Merali2011,Modi2011}. While entanglement captures
the non-separability of two subsystems \cite{Plenio2007,Horodecki2009},
quantum discord detects nonclassical properties even in separable
states. Different attempts were presented to connect the new concept
of quantum discord to quantum entanglement
\cite{Madhok2011,Cavalcanti2011,Streltsov2011,Piani2011,Cornelio2011,Fanchini2011,Piani2011a},
and to broadcasting \cite{Piani2008,Piani2009,Luo2010b}. Several experimental results
have  been reported in \cite{Lanyon2008,Xu2010,Auccaise2011,Passante2011}. Quantum
discord, as well as related quantifiers of quantum correlations
\cite{Modi2011,Oppenheim2002,Rajagopal2002,Luo2008,Wu2009,Modi2010,Dakic2010,Rossignoli2010,Streltsov2011,Piani2011,Streltsov2011a,Giorgi2011a,Lang2011, Xu2011},
have also  been linked to better-than-classical performance in quantum computation
and communication tasks, even in the presence of limited or strictly
vanishing entanglement \cite{Knill1998,DiVincenzo2004,Datta2008,Lanyon2008,Brodutch2011,Chaves2011,Eastin2010,Modi2010a,Madhok2011a,Boixo2011}.
 An important question to understand the role of quantum correlations as signatures of genuine
nonclassical  behavior is whether they distribute
in a monogamous way among multipartite systems.

A bipartite measure of correlations ${\cal Q}$ satisfies {\em monogamy}
if \cite{Coffman2000,Horodecki2009}
\begin{equation}
{\cal Q}^{A|BC}(\rho_{ABC})\geq{\cal Q}^{A|B}(\rho_{AB})+{\cal Q}^{A|C}(\rho_{AC})\label{eq:monogamy}
\end{equation}
holds for all states $\rho_{ABC}$.
Here $\rho_{AB}=\tr_C(\rho_{ABC})$ denotes the reduced state of parties $A$ and $B$, and
analogously for   $\rho_{AC}$. The vertical bar is the familiar notation for the bipartite
split. The concept of monogamy is visualized in Fig. \ref{fig:1}.

\begin{figure}
\begin{centering}
\includegraphics[width=0.5\columnwidth]{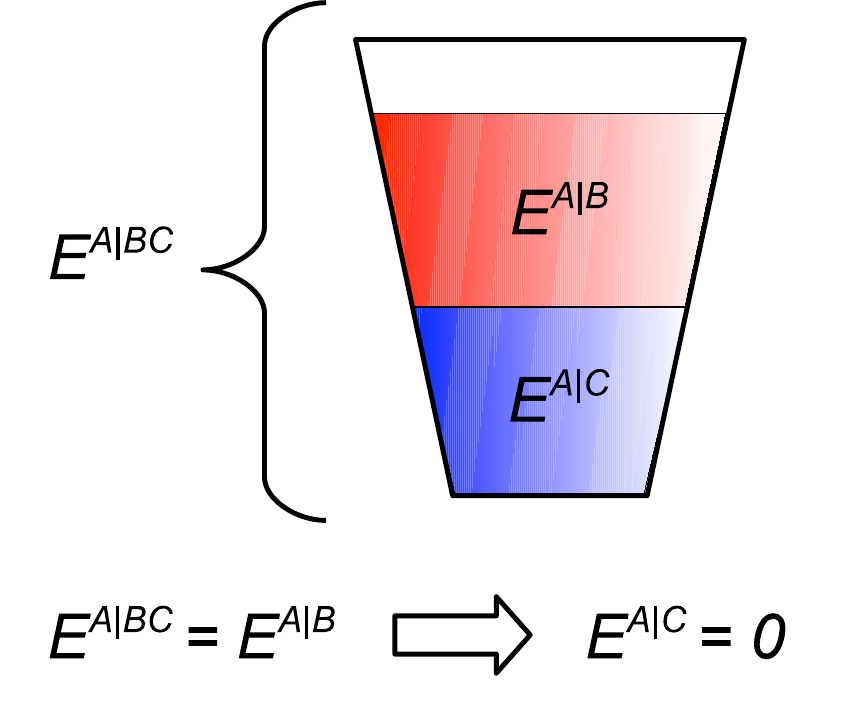}
\par\end{centering}

\caption{\label{fig:1}[Color online] Entanglement is monogamous: for a fixed amount of entanglement between $A$ and $BC$, the more entanglement exists between $A$ and $B$, the less can exist between $A$ and $C$. Quantitatively
this is expressed using the monogamy relation, see Eq. \eqref{eq:monogamy}
in the main text. In particular, the latter implies---for a monogamous measure of entanglement $E$---that $E^{A|C}=0$ if  $E^{A|BC}=E^{A|B}$. In this Letter we show that the monogamy relation
does not hold in general for any quantum correlation measure beyond entanglement, i.e. for any measure
that does not vanish on  separable states.}
\end{figure}

If ${\cal Q}$ denotes in particular an entanglement measure \cite{Plenio2007,Horodecki2009},
then there are a number of choices that satisfy monogamy for pure states of qubits, including
the squared concurrence \cite{Coffman2000}, and the squared negativity
\cite{Ou2007}, as well as their continuous variable counterparts
for multimode Gaussian states \cite{Adesso2006,Hiroshima2007}. The only known measure that is monogamous in all dimensions is the squashed entanglement  \cite{Christandl2004,Koashi2004}.
Other entanglement measures such as, e.g., the entanglement of formation do not satisfy
the monogamy relation  \cite{Coffman2000}. There is no known {\em a priori} rule
about whether a given entanglement measure is monogamous or not.
It is natural to ask whether a given measure for general quantum correlations is
monogamous. Certain
measures of general quantum correlations, such as quantum discord,
were shown to violate monogamy by finding explicit examples
of states for which the inequality \eqref{eq:monogamy} does not hold \cite{Giorgi2011,Prabhu2011,Prabhu2011a,Sudha2011,Allegra2011,Ren2011}.
Those examples, however, do not  exclude the possibility that other measures of
quantum correlations, akin to the quantum discord, could exist
that do satisfy a monogamy inequality.

In this Letter we address the issue of whether monogamy, in general, can extend to general quantum correlations beyond entanglement. Quantitatively this question
can be formulated as follows:
\emph{Does there exist a measure of
correlations $\cal{Q}$ which obeys the monogamy relation \eqref{eq:monogamy} and is nonzero on a separable state?}
We will put this question to rest by proving that
all measures for quantum correlations beyond entanglement
(i.e., that are non-vanishing on at least some separable state)
and that respect some basic properties are {\em not} monogamous in
general. These basic properties of the correlation measure $\cal{Q}$
are the following:
\begin{itemize}
\item positivity, i.e.
\begin{equation}
\label{eq:positivity}
{\cal Q}^{A|B}\left(\rho_{AB}\right)\geq 0;
\end{equation}
\item invariance under local unitaries $U_A\otimes V_B$, i.e.
\begin{equation}
\label{eq:LUinvariance}
{\cal Q}^{A|B}\left(\rho_{AB}\right) = {\cal Q}^{A|B}\left(U_A\otimes V_B \rho_{AB} U^\dagger_A\otimes V^\dagger_B\right);
\end{equation}
\item no-increase upon attaching a local ancilla, i.e.
\begin{align}
{\cal Q}^{A|B}\left(\rho_{AB}\right) & \geq{\cal Q}^{A|BC}\left(\rho_{AB}\otimes\ket{0}\bra{0}_{C}\right).\label{eq:ancilla}
\end{align}
\end{itemize}
These properties are valid for several measures of
correlations known in the literature, including all entanglement measures~\cite{Plenio2007,Horodecki2009}. In particular, positivity and invariance under local unitaries are standard requirements \cite{Brodutch2011a}. For the quantum discord
defined in Ref. \cite{Ollivier2001,Henderson2001}, which is an asymmetric quantity, Eq.~\eqref{eq:ancilla}
can be verified by inspection, and is valid independently of whether the ancilla is attached on the side where the measurement entering the definition of discord is to be performed, or on the unmeasured side. In a more general
scenario, quantum correlations can be defined as the minimal distance
to the set of classically correlated states \cite{Modi2010,Dakic2010,Piani2011,Streltsov2011a}.
In this case Eq.~\eqref{eq:ancilla} follows
from the fact that any "reasonable" distance does not change upon attaching an
ancilla: $D\left(\rho,\sigma\right)=D\left(\rho\otimes\ket{0}\bra{0},\sigma\otimes\ket{0}\bra{0}\right)$.
The same arguments can be applied to measures which are defined via
measurements on local subsystems \cite{Luo2008}. Alternatively, quantum correlations may be investigated and quantified in terms of the minimal
amount of entanglement necessarily created between the system and a measurement
apparatus realizing a complete projective measurement~\cite{Streltsov2011,Piani2011,Piani2011a,Gharibian2011}. Eq. \eqref{eq:ancilla} also holds in this case, which can be seen
solely using the properties of entanglement measures.

We are now in position
to prove the following theorem.
\begin{thm}
\label{thm:1}A measure of correlations ${\cal Q}$ that respects Eqs. \eqref{eq:positivity},
\eqref{eq:LUinvariance}, and \eqref{eq:ancilla}, and is also monogamous according to \eqref{eq:monogamy}, must vanish for all separable states.\end{thm}
\begin{proof}
Consider a measure $\mathcal{Q}$ respecting the hypothesis, and a generic separable state $\rho_{AC}=\sum_{i}p_{i}\ket{\psi_{i}}\bra{\psi_{i}}_{A}\otimes\ket{\phi_{i}}\bra{\phi_{i}}_{C}$.
In the following we will concentrate on a special extension of $\rho_{AC}$,
defined as
\begin{equation}
\rho_{ABC}=\sum_{i}p_{i}\ket{\psi_{i}}\bra{\psi_{i}}_{A}\otimes\ket{i}\bra{i}_{B}\otimes\ket{\phi_{i}}\bra{\phi_{i}}_{C},\label{eq:rho}
\end{equation}
with orthogonal states $\left\{ \ket{i}_{B}\right\} $. Observe that
$\rho_{ABC}$ has the same amount of correlations ${\cal Q}^{A|BC}$ as
the state
\begin{equation}
\sigma_{ABC}=\sum_{i}p_{i}\ket{\psi_{i}}\bra{\psi_{i}}_{A}\otimes\ket{i}\bra{i}_{B}\otimes\ket{0}\bra{0}_{C},
\label{eq:sigma}
\end{equation}
since both states are related by a local unitary on $BC$. On the
other hand, Eq. \eqref{eq:ancilla} implies that $\sigma_{ABC}$ does not have more correlations than the reduced state $\sigma_{AB}$. Taking these
two observations together we obtain ${\cal Q}^{A|B}\left(\sigma_{AB}\right)\geq{\cal Q}^{A|BC}\left(\rho_{ABC}\right)$.
Now we invoke the monogamy relation for the state $\rho_{ABC}$, which leads
us to the inequality
\begin{equation}
{\cal Q}^{A|B}\left(\sigma_{AB}\right)\geq{\cal Q}^{A|B}\left(\rho_{AB}\right)+{\cal Q}^{A|C}\left(\rho_{AC}\right).
\end{equation}
The final ingredient in the proof is the fact that the two states
$\rho_{AB}$ and $\sigma_{AB}$ are equal. From the positivity of the measure it follows immediately
that ${\cal Q}^{A|C}$ must vanish on the state $\rho_{AC}$. Since the latter is a generic separable state, ${\cal Q}$ must vanish on all separable states.
\end{proof}

The power of Theorem \ref{thm:1} lies in its generality. Under very weak assumptions
it rules out the existence of monogamous correlations beyond entanglement.
Note that the arguments used in the proof of Theorem \ref{thm:1} are strong
enough to show that the violation of monogamy appears even in three-qubit
systems. This can be seen starting from Eq. (\ref{eq:rho}), with each subsystem
being a qubit. The measure ${\cal Q}$ violates monogamy, if it is nonzero on some separable two-qubit state of rank two.
This is the case for quantum discord and any related measures of quantum
correlations.

As we have argued below Eq. (\ref{eq:ancilla}), the properties (\ref{eq:positivity}-\ref{eq:ancilla}) are satisfied by all reasonable measures of quantum correlations known to the authors. However, in general it cannot be excluded that the measure under study violates one of the properties given in Eq. (\ref{eq:positivity}),
(\ref{eq:LUinvariance}), or (\ref{eq:ancilla}). Alternatively, we assume that some of these properties cannot be proven. In this situation, Theorem \ref{thm:1} does not tell us whether ${\cal Q}$ is monogamous
or not. Then, it is still possible to show that a monogamous measure ${\cal Q}$ must be zero
on all separable states, if it remains finite for a fixed dimension
of one subsystem, i.e. if 
\begin{equation}
\mathcal{Q}^{A|B}\leq f\left(d_{A}\right)<\infty\label{eq:bounded}
\end{equation}
for fixed $d_{A}$, and some function $f$. To see this we use the fact that any separable
state $\rho_{AB}$ has a symmetric extension $\rho_{AB_{1}\cdots B_{n}}$
such that $\rho_{AB}=\rho_{AB_{i}}$ holds for all $1\leq i\leq n$, where $n$ is an arbitrary positive integer
\cite{Werner1989,Werner1989a,Doherty2004,Dong2006}. Eq. (\ref{eq:bounded}) implies that the measure ${\cal Q}^{A|B_{1}\cdots B_{n}}\left(\rho_{AB_{1}\cdots B_{n}}\right)$ is finite for all $n$, including the limit $n \rightarrow \infty$. On the other hand, if $\cal Q$ is monogamous, it has to fulfill the following inequality: 
\begin{equation}
{\cal Q}^{A|B_{1}\cdots B_{n}}\left(\rho_{AB_{1}\cdots B_{n}}\right)\geq n{\cal Q}^{A|B}\left(\rho_{AB}\right).\label{eq:Werner}
\end{equation}
However, if the measure ${\cal Q}$ is nonzero on the separable state $\rho_{AB}$, one can always choose some $n$ which is large enough such that Eq. (\ref{eq:Werner})
is violated, and thus ${\cal Q}$ cannot be monogamous.

So far we have presented two different ways to show that a given measure of quantum correlations $\cal{Q}$ violates monogamy, namely Theorem \ref{thm:1} and Eq. (\ref{eq:bounded}). At this stage it is natural to ask whether these two results have the same power, i.e. whether they allow to draw the same conclusions about the structure of a given measure $\cal{Q}$. As already noted above, the proof of Theorem \ref{thm:1} allows to rule out monogamy even for the simplest case of three qubits, as long as the measure $\cal{Q}$ does not vanish on some separable state of two qubits having rank not larger than two. On the other hand, this argument does not apply to Eq. (\ref{eq:bounded}) and (\ref{eq:Werner}). Indeed, if $\cal{Q}$ is nonzero on some separable two-qubit state $\rho_{AB}$, Eq. (\ref{eq:bounded}) and (\ref{eq:Werner}) only allow the statement that the measure $\cal{Q}$ violates monogamy for some extension $\rho_{AB_1\ldots B_n}$. In particular, if $n>2$, this result does not provide any insight about the monogamy of the measure for three-qubit states.

We move on to observe that monogamy (Eq. \eqref{eq:monogamy}), together with positivity (Eq. \eqref{eq:positivity}), invariance under local unitary (Eq. \eqref{eq:LUinvariance}) and no-increase under attaching a local ancilla (Eq. \eqref{eq:ancilla}), imply no-increase under local operations. This is due to the fact that any quantum operation $\Lambda$ admits a Stinespring dilation:
$\Lambda[\rho_B]=\tr_C\left(U_{BC}\rho_B\otimes \proj{0}_C  U_{BC}^\dagger\right)$, i.e. any quantum
operation can be seen as resulting from a unitary operation on a larger-dimensional Hilbert space.
Thus, for ${\cal Q}$ respecting  Eqs. \eqref{eq:monogamy}, \eqref{eq:positivity}, \eqref{eq:LUinvariance}, and \eqref{eq:ancilla}, one finds
\begin{eqnarray}
{\cal Q}^{A|B}\left(\rho_{AB}\right)
&\geq &
{\cal Q}^{A|BC}\left(\rho_{AB}\otimes\ket{0}\bra{0}_{C}\right)\nonumber \\
&=&
{\cal Q}^{A|BC}\left(U_{BC} \rho_{AB}\otimes\ket{0}\bra{0}_{C}U_{BC}^\dagger \right)\nonumber\\
&\geq&
{\cal Q}^{A|B}\left(\tr_C\left(U_{BC} \rho_{AB}\otimes\ket{0}\bra{0}_{C}U_{BC}^\dagger \right)\right)\nonumber\\
&\quad+&
{\cal Q}^{A|C}\left(\tr_B\left(U_{BC} \rho_{AB}\otimes\ket{0}\bra{0}_{C}U_{BC}^\dagger \right)\right)\nonumber\\
&\geq&
{\cal Q}^{A|B}\left(\Lambda_B[\rho_{AB}]\right).
\end{eqnarray}
No-increase under local operations~\footnote{To be precise, no-increase just under operations on one side.}, and thus, \emph{a fortiori}, monogamy (the latter together with the almost trivial properties \eqref{eq:positivity}, \eqref{eq:LUinvariance}, and \eqref{eq:ancilla}) imply the following

\begin{thm}
A measure of correlations ${\cal Q}$ that is non-increasing under operations on at least one side must be maximal on pure states; that is, for any $\rho_{AB}$ on $\mathbb{C}^d\otimes \mathbb{C}^d$ there exists a pure state $\proj{\psi}_{AB}\in\mathbb{C}^d\otimes \mathbb{C}^d$ such that ${\cal Q}^{A|B}(\proj{\psi}_{AB})\geq {\cal Q}^{A|B}(\rho_{AB})$.\end{thm}
\begin{proof}
Immediate when one uses the fact that any state $\rho_{AB}$ can be seen as the result of the application of a channel $\Lambda_B$ ($\Lambda_A$) on any purification $\ket{\psi}_{AB}$ of $\rho_{A}$ ($\rho_B$) (see, for example,~\cite{Christandl2004}). Suppose that the measure ${\cal Q}$ is non-increasing under quantum operations  on $A$. Then:
\begin{equation}
{\cal Q}^{A|B}(\proj{\psi}_{AB})\geq {\cal Q}^{A|B}(\Lambda_A[\proj{\psi}_{AB}])={\cal Q}^{A|B}(\rho_{AB}).
\end{equation}

\end{proof}
This simple theorem is relevant, in particular, for the case of \emph{symmetric} measures of quantum correlations. Several
such measures were proposed in Refs. \cite{Modi2010,Piani2011,Streltsov2011a}.
Some of these measures have counterintuitive properties. In particular, in  \cite{Piani2011} it was shown
that for the relative entropy of quantumness
there exist mixed states $\rho_{AB}$ which have more quantum
correlations than any pure state $\ket{\psi_{AB}}$. The just proven theorem can be interpreted as a signature of the fact that general quantum correlations can increase under local operations (and \emph{a fortiori} as a signature of the lack of monogamy) \cite{Streltsov2011a}.

Theorem \ref{thm:1} and the reasoning in its proof amount essentially to the following insight about the violation of monogamy:
if there is a separable state $\rho_{AB}$ with nonzero correlations ${\cal Q}$, then there exists a \emph{mixed}
state $\rho_{ABC}$ which proves that the measure under scrutiny is not monogamous: ${\cal Q}^{A|BC}\left(\rho_{ABC}\right)<{\cal Q}^{A|B}\left(\rho_{AB}\right)+{\cal Q}^{A|C}\left(\rho_{AC}\right)$.
On the other hand, crucially, a measure of correlations can still respect monogamy when evaluated on pure states
$\rho_{ABC}=\ket{\psi}\bra{\psi}_{ABC}$. As will be demonstrated
in the following, the geometric measure of discord has exactly this
property for three qubits. Before we present this result, we recall the definition
of this measure.

The geometric measure of discord $D_G$ was defined in Ref. \cite{Dakic2010}
as the minimal Hilbert-Schmidt distance to the set of classical-quantum states (CQ):
\begin{equation}
D_{G}^{A|B}\left(\rho_{AB}\right)=\min_{\sigma_{AB}\in CQ}\left\Vert \rho_{AB}-\sigma_{AB}\right\Vert ^{2}_2.\label{eq:D}
\end{equation}
Here we used the 2-norm, also known as Hilbert-Schmidt norm, $\left\Vert \rho-\sigma\right\Vert _2=\sqrt{\mathrm{Tr}\left(\rho-\sigma\right)^{2}}$,
and the minimum is taken over all classical-quantum states $\sigma_{AB}$.
These are states which can be written as $\sigma_{AB}=\sum_{i}p_{i}\ket{i}\bra{i}_{A}\otimes\sigma_{B}^{i}$
with some local orthogonal basis $\left\{ \ket{i_{A}}\right\} $. The geometric discord has an operational interpretation in terms of the average fidelity of the remote state preparation protocol for two-qubit systems \cite{Dakic2012}. As noted above,
the geometric measure of discord cannot be monogamous in general,
since it is nonzero on some separable states. However, this measure
is monogamous for all pure states of three qubits.
\begin{thm}
\label{thm:3}The geometric measure of discord is monogamous for all
pure states $\ket{\psi}_{ABC}$ of three qubits:
\begin{equation}
D_{G}^{A|BC}\left(\ket{\psi}\bra{\psi}_{ABC}\right)\geq D_{G}^{A|B}\left(\rho_{AB}\right)+D_{G}^{A|C}\left(\rho_{AC}\right),\label{eq:monogamy-pure}
\end{equation}
where $\rho_{AB}=\mathrm{Tr}_C(\ket{\psi}\bra{\psi}_{ABC})$ and analogously for $\rho_{AC}$.\end{thm}
\begin{proof}
We notice that for proving the inequality in Eq. \eqref{eq:monogamy-pure}
it is enough to show that for any pure state $\ket{\psi}_{ABC}$ there exists
a classical-quantum state $\sigma_{ABC}$ such that
\begin{equation}
D_{G}^{A|BC}\left(\ket{\psi}\bra{\psi}_{ABC}\right)\geq\left\Vert \rho_{AB}-\sigma_{AB}\right\Vert ^{2}_2+\left\Vert \rho_{AC}-\sigma_{AC}\right\Vert ^{2}_2.\label{eq:monogamy-pure-1}
\end{equation}
This inequality then automatically implies inequality \eqref{eq:monogamy-pure},
as, due to the minimization in the geometric measure of discord,
the right-hand side of \eqref{eq:monogamy-pure} can only be smaller than or equal to the
right-hand side of \eqref{eq:monogamy-pure-1} .
In order to show the existence of the mentioned  classical-quantum state $\sigma_{ABC}$
we choose a specific parametrization
for a pure state of three qubits \cite{Brun2001}:
\begin{eqnarray}
\ket{\psi_{ABC}} & = & \sqrt{p}\ket{0}_{A}\left(a\ket{00}_{BC}+\sqrt{1-a^{2}}\ket{11}_{BC}\right)\nonumber \\
 &  & +\sqrt{1-p}\ket{1}_{A}\bigg[\gamma\big(\sqrt{1-a^{2}}\ket{00}_{BC}-a\ket{11}_{BC}\big)\nonumber \\
 &  & +f\ket{01}_{BC}+g\ket{10}_{BC}\bigg].
\end{eqnarray}
The real numbers $p$, $a$ and $f$ range between $0$ and $1$,
$g$ is complex with $0\leq f^{2}+\left|g\right|^{2}\leq1$, and $\gamma=\sqrt{1-f^{2}-\left|g\right|^{2}}$ is also real.

We proceed by evaluating the left-hand side of Eq. \eqref{eq:monogamy-pure-1},
using the explicit formula for pure states \cite{Luo2011,Girolami2011b}:
\begin{equation}
D_{G}^{A|BC}\left(\ket{\psi}\bra{\psi}_{ABC}\right)=2\left(1-p\right)p.
\end{equation}
In the next step we define the classical-quantum state $\sigma_{ABC}=\sum_{i=0}^{1}\Pi_{A}^{i}\rho_{ABC}\Pi_{A}^{i}$
with local projectors in the computational basis: $\Pi_{A}^{i}=\ket{i}\bra{i}_{A}$.
The evaluation of the right-hand side of Eq. \eqref{eq:monogamy-pure-1}
is straightforward:
\begin{equation}
\left\Vert \rho_{AB}-\sigma_{AB}\right\Vert ^{2}_2+\left\Vert \rho_{AC}-\sigma_{AC}\right\Vert ^{2}_2=2c\left(1-p\right)p
\end{equation}
 with $c=1+[4a^{2}\left(1-a^{2}\right)-1] \gamma^{2}$.
The proof is complete, if we can show that $c$ cannot be larger than
$1$. This can be seen by noting that the term $4a^{2}\left(1-a^{2}\right)$
is maximal for $a^{2}=\frac{1}{2}$, which leads to the maximal possible
value $c=1$.
\end{proof}

Even though quantum correlations beyond entanglement cannot be monogamous
in general, Theorem \ref{thm:3} demonstrates that for pure states of three qubits
monogamy of the geometric measure of discord is still preserved. To the best of our knowledge this is the first instance of a measure of quantum correlations beyond entanglement that satisfies a restricted monogamy inequality.
Certainly, this is {\em not} a  property which all measures of quantum correlations have in common:
As shown e.g. in Ref. \cite{Giorgi2011},  the original quantum discord
violates monogamy even on some pure states of three qubits.

In conclusion, we have addressed the question of monogamy for quantum correlations beyond entanglement. Using very general arguments, we have proven that any measure of correlations which is nonzero on some separable state unavoidably violates monogamy.  Furthermore, we have shown that any monogamous  measure of quantum correlations must be maximal on pure states. These results imply severe constraints on any monogamous measure of quantum correlations, and can also be used to witness the violation of monogamy. Finally we have shown  that even though all measures of nonclassical correlations akin to quantum discord cannot be monogamous for all states, they still may obey monogamy in certain restricted situations. In particular, we proved that the geometric measure of discord is monogamous for all pure states of three qubits. It is an open question whether there exists a measure of general quantum correlations which is monogamous for tripartite pure states of arbitrary dimensions. Another open question, which points to a possible future research direction, arises from the generalization of quantum discord to theories which are more general than quantum \cite{Perinotti2012}. We hope that the results presented in this paper are also useful for this more general scenario. - Thus, the answer to the question posed in the title is: General quantum correlations are in general not monogamous.

\emph{Acknowledgements:} We thank Davide Girolami and Hermann Kampermann for discussions.
MP acknowledges support by NSERC, CIFAR, Ontario Centres of Excellence. GA is supported by a Nottingham Early Career Research and Knowledge Transfer Award. MP and GA acknowledge joint support by the EPSRC Research Development Fund (Pump Priming grant 0312/09). GA acknowledges ESF for sponsoring the workshop during which this work was started. DB and AS acknowledge financial support by DFG; AS was supported by ELES.

\bibliographystyle{apsrev4-1}
\bibliography{literature}

\end{document}